\documentclass{article}
\usepackage{times}  
\usepackage{helvet}  
\usepackage{courier}  
\usepackage[hyphens]{url}  
\usepackage{graphicx} 
\usepackage{caption} 
\usepackage[letterpaper, margin=1in]{geometry}

\usepackage[linesnumbered,ruled,vlined]{algorithm2e}
\usepackage{natbib}
\SetKwInput{KwInput}{Input}                
\SetKwInput{KwOutput}{Output} 
\SetKwInput{KwPre}{Precondition}
\SetKwInput{KwInv}{Invariant}
\SetKwFunction{FPTAS}{{\sc FPTAS}} 
\SetKwFunction{FLocalSearch}{{\sc LocalSearch}}
\SetKwFunction{FDynamicProgram}{{\sc DynamicProgram}}
\SetKwFunction{Update}{{\sc Update}}

\usepackage{amsmath,amsfonts,amssymb,bbm,amsthm,bm,xspace}
\usepackage{color}
\usepackage{thm-restate}
\usepackage{authblk}

\theoremstyle{definition}
\newtheorem{definition}{Definition}[section]
\theoremstyle{plain}
\newtheorem{corollary}{Corollary}[section]
\newtheorem{theorem}{Theorem}[section]

\newtheorem{proposition}[theorem]{Proposition}

\newcommand{\EFP}{{\sc EFprior}\xspace}

\title{Fair Division with Prioritized Agents}
\author[1]{Xiaolin Bu}
\author[2]{Zihao Li}
\author[3]{Shengxin Liu}
\author[4]{Jiaxin Song}
\author[5]{Biaoshuai Tao}
\affil[1]{Shanghai Jiao Tong University, lin\underline{ }bu@sjtu.edu.cn}
\affil[2]{Nanyang Technological University, zihao004@e.ntu.edu.sg}
\affil[3]{Harbin Institute of Technology, Shenzhen, sxliu@hit.edu.cn}
\affil[4]{Shanghai Jiao Tong University, sjtu\underline{ }xiaosong@sjtu.edu.cn}
\affil[5]{Shanghai Jiao Tong University, bstao@sjtu.edu.cn}
\date{}

\begin{document}

\maketitle

\begin{abstract}
We consider the fair division problem of indivisible items. It is well-known that an envy-free allocation may not exist, and a relaxed version of envy-freeness, \emph{envy-freeness up to one item (EF1)}, has been widely considered. In an EF1 allocation, an agent may envy others' allocated shares, but only up to one item. In many applications, we may wish to specify a subset of prioritized agents where strict envy-freeness needs to be guaranteed from these agents to the remaining agents, while ensuring the whole allocation is still EF1. Prioritized agents may be those agents who are envious in a previous EF1 allocation, those agents who belong to underrepresented groups, etc. Motivated by this, we propose a new fairness notion named \emph{envy-freeness with prioritized agents} \EFP, and study the existence and the algorithmic aspects for the problem of computing an \EFP allocation. With additive valuations, the simple round-robin algorithm is able to compute an \EFP allocation. In this paper, we mainly focus on general valuations. In particular, we present a polynomial-time algorithm that outputs an \EFP allocation with most of the items allocated. When all the items need to be allocated, we also present polynomial-time algorithms for some well-motivated special cases.
\end{abstract}

\section{Introduction}
The \emph{fair division} problem studies how to \emph{fairly} allocate a set of resources to a set of agents who have heterogeneous preferences over the resources.
Starting with \cite{Steinhaus1948}, the fair division problem has been receiving significant attention from mathematicians, economists, and computer scientists in the past decades.
Among different fairness interpretations, \emph{envy-freeness}~\citep{Foley67} is the most studied fairness criterion which requires that each agent believes she receives a share that has weakly more value than the share allocated to each of the other agents (i.e., each agent does not envy any other agents).
Classical work in fair division has been focused on resources that are \emph{infinitely divisible}, which is also known as \emph{cake-cutting} problem~\citep{even1984note,brams1995envy,CL10,bei2012optimal,bei2017cake,10.1145/3490486.3538321} . Envy-free allocations always exist in the cake-cutting setting~\citep{brams1995envy} and can be computed via a discrete and bounded protocol~\citep{AzizMa16}.

Recent research focuses more on allocations of indivisible items \citep{Lipton04onapproximately, budish2011combinatorial,ConitzerFrSh17,CaragiannisKuMo19,bu2022complexity,amanatidis2022fair,li2022proportional}.
Obviously, violation of fairness is unavoidable in some scenarios, e.g., when the number of items is less than the number of agents (in which case some agents will receive an empty set).
This necessitates the relaxation of fairness.
To relax envy-freeness, \citet{Lipton04onapproximately} and \citet{budish2011combinatorial} proposed the notion \emph{envy-freeness up to one item (EF1)} which allows an agent $i$ to envy another agent $j$, as long as there exists an item in $j$'s allocated bundle whose (hypothetical) removal eliminates the envy from $i$ to $j$.\footnote{EF1 was implicitly mentioned by \citet{Lipton04onapproximately}, and explicitly formulated by \citet{budish2011combinatorial}.}
EF1 has then been one of the most widely-studied fairness notions for indivisible item allocation, and is guaranteed to exist~\citep{Lipton04onapproximately,CaragiannisKuMo19}.

While an allocation that gives advantages to some agents over the others is perhaps justifiable by the inherent unfairness in the allocation of indivisible items, in many applications, it is desirable that a specified set of agents with higher priority are favored.
The examples abound: it is natural to prioritize those agents who are not favored in the past allocations (due to the intrinsic unfairness of item-allocation); for applicants with equal qualifications, job positions are given to those applicants in underrepresented groups first.
These motivate the proposal of new fairness solution concepts that not only mitigates the unfairness but also ensures that those prioritized agents are favored if unfairness is inevitable. 

In the context of envy-freeness, although EF1 mitigates the unfairness by restricting envy to ``up to one item'', it does not consider agents' priorities.
To incorporate this feature, we propose a new fairness notion called \emph{envy-freeness with prioritized agents} (\EFP).
Given a set of prioritized agents, an \EFP allocation requires that 1) the allocation as a whole is EF1 and 2) strict envy-freeness from each prioritized agent to each non-prioritized agent is ensured.

As an important remark, by introducing prioritized agents, our solution concept \EFP does not \emph{create} unfairness.
Instead, we are seeking for allocations that prioritize a prescribed set of agents \emph{subject to that} unfairness has been justified and mitigated.
This differentiates our work from those who assign weights to agents based on their importance~\citep{ChakrabortyIgSu21,ChakrabortySeSu22}.
Our solution concept applies to the scenarios where fairness is of paramount importance and the ``tie-breaking'' matters if absolute fairness fails.

\paragraph{Applications of fairness with prioritized agents.}
Our model aligns with all the fair division applications where fairness is the primary desideratum, and a secondary desideratum is used to break ties.
This secondary desideratum can be arbitrarily specified as needed.
For example, in a typical course registration mechanism of a university, students submit rankings to the available courses reflecting their preferences, and the system allocates courses based on their (heterogeneous) preferences and course vacancies.
The course allocation is primarily based on students' preferences and will allocate a course to a student only when all the students ranking this course higher than her are registered.
This minimizes the envy among the students.
However, when the number of vacancies cannot sustain all the students who rank this course high, priority is given to senior students among those students submitting the same ranking, in order to ensure these senior students' timely graduation.
Here, the secondary desideratum takes the academic year into account.
As another example, when a set of entitled benefits cannot be evenly distributed among the employees in a company, slight advantages are normally given to those older employees for their longer time of service, those who have larger family expense (e.g., having more children, being under medical treatment), or other desirable tie-breaking factors. 

Another potential application of our model is the fair division with underrepresented agents.
A vast among of organizations value DEI (diversity, equity, and inclusion) and offer corresponding training program to avoid discrimination to those underrepresented groups.
However, the mere presence of these ``diversity structures'' may fail to serve their purpose without a concrete fairness measurement~\citep{kaiser2013presumed}.
On the other hand, this type of measurements needs to be carefully made to avoid \emph{reverse discrimination}: inappropriate  policies may cause those ``over-represented groups'' feel they have been discriminated~\citep{fish1993reverse,newkirk2017myth}.
Our \EFP notion gives a concrete measurement and provides advantages to the underrepresented groups to an extent that is also acceptable to other groups in the sense that an overall fairness criterion EF1 is still guaranteed.
The practice of prioritizing underrepresented groups within the range of fairness has already been adopted widely.
For example, under the Equality Act 2010 in the United Kingdom, the membership in a protected and disadvantaged group is allowed to be considered in hiring and promotion, if the candidates are of equal merit. In this case, the membership in an underrepresented group is used as a ``tie-breaker''.

\paragraph{On computing \EFP allocations.}
Unfortunately, most EF1 algorithms do not have control over which agents are favored in the output EF1 allocation.
\citet{Lipton04onapproximately} proposed an algorithm, \emph{envy-graph procedure}, that computes an EF1 allocation with polynomial time.
The algorithm starts by assigning each agent the empty bundle, and adds an item to an agent's bundle in each iteration.
Specifically, an \emph{envy-graph} (see Sect.~\ref{sect:prelim-cycle} for details) is constructed where vertices represent agents and a directed edge $(i,j)$ represents $i$ envies $j$ in the current allocation.
To maintain EF1 property throughout the process, the algorithm always chooses a source vertex in the graph and adds an item to the bundle of the agent corresponding to this vertex.
Whenever there is a cycle in the graph, a \emph{cycle-rotation} step is performed where each agent's bundle is replaced by the bundle of the next agent in the cycle; this guarantees the existence of the source vertices at the beginning of each iteration.
In this algorithm, a source in the envy-graph is not envious and is thus favored, but the ``cycle-rotation'' step in the algorithm changes the set of source vertices unpredictably.

\citet{CaragiannisKuMo19} showed that the allocation with the maximum Nash social welfare (NSW) is always EF1.
Thus, a natural algorithm for computing an EF1 allocation is to find a NSW maximizing allocation, which is adopted by the fair division website \emph{spliddit.org}~\citep{goldman2015spliddit,shah2017spliddit}.
However, the uniqueness of the NSW maximizing allocations makes it impossible to prioritize a fixed set of agents.

Although a simple round-robin algorithm~\citep{CaragiannisKuMo19} can output an EF1 allocation with specified agents prioritized, it only works if agents' valuations are additive (see Sect.~\ref{sect:additive} for details).
However, in general settings, the existence and the computation of an \EFP allocation remain to be open problems, which is the main concern of this paper.

\subsection{Our Results}
As our main contribution, we propose a new fairness notion \EFP that is stronger than EF1 by additionally enforcing strict envy-freeness from a prescribed set of prioritized agents to the remaining agents.
We then study the existence and the algorithmic aspects of \EFP.
We note that \EFP always exists for agents with additive valuations, and can be computed by a round-robin algorithm (Sect.~\ref{sect:additive}).

With general valuations, we present a polynomial-time algorithm that outputs a partial \EFP allocation where the set of unallocated items has a small bounded size and small bounded values to all agents (Theorem~\ref{thm:efppar}).
Our techniques are built upon the algorithm for computing an \emph{envy-freeness up to any item (EFX)} allocation proposed by~\citet{charity_soda}.
Other than some additional effort being made to maintaining the strict envy-freeness from the prioritized agents to the remaining agents, our algorithm makes use of a novel approach to exchange the agents' bundles with the pool of unallocated items.
This new approach makes our algorithm run in polynomial time, whereas Chaudhury et al.'s algorithm is only known to run in pseudo-polynomial time.

We also study two special cases: when all the prioritized agents have the same valuations, and when all the non-prioritized agents have the same valuations (Sect.~\ref{sect:specialgeneral}).
Under both cases, we show the existence of \EFP by presenting polynomial time algorithms.
The positive results on these two special cases imply the tractability of \EFP in the following three scenarios: 1) when there is only one prioritized agent, 2) when there is only one non-prioritized agent, and 3) when the total number of agents is at most 3.
The algorithms in Sect.~\ref{sect:specialgeneral} provide some basic ideas that are used in the algorithm for our main technical result (mentioned in the previous paragraph) in Sect.~\ref{sect:general}.

We conclude our paper by proposing the open problem about the existence and the computational complexity of a (complete) \EFP allocation in the general setting.

\subsection{Related Work}
Besides envy-freeness and EF1 as introduced before, another important envy-based fairness notion is called \emph{envy-freeness up to any item (EFX)}~\citep{CaragiannisKuMo19}.
In an EFX allocation, each agent $i$ may envy agent $j$ but the envy can be eliminated by removing an \emph{arbitrary} item from agent $j$'s bundle.
Clearly, EFX is a stronger fairness notion than EF1.
The existence of EFX allocations is largely open. 
We only know that EFX allocations exist in some special cases, e.g., two agents with general valuations~\citep{PlautRo20} and three agents with additive valuations~\citep{ChaudhuryGaMe20}.
Other notable indivisible fairness notions include \emph{proportionality up to one item (PROP1)}~\citep{ConitzerFrSh17}, \emph{maximin share (MMS)} fairness~\citep{budish2011combinatorial}, and so on. 
These notions are also extended to incorporate agents' weights or entitlements such as weighted EF1~\citep{ChakrabortyIgSu21}, weighted PROP1~\citep{AzizMoSa20} and weighted MMS~\citep{FarhadiGhHa19}. 

Previous work also considers partial allocations. \citet{CaragiannisGrHu19} show that there exist partial EFX allocations that have high Nash social welfares. \citet{charity_soda} show that partial EFX allocations exist if the number of unallocated goods is at most $n-1$. Furthermore, \citet{ChaudhuryGaMe21} prove that there always exists a $(1-\epsilon)$-EFX allocation with $64(n/\epsilon)^{4/5}$ unallocated goods and high Nash welfare. The bound of the size of unallocated goods has been respectively improved to 
$O\left(n^{0.67}\right)$ for any $\epsilon \in (0, \frac12]$ by \cite{berendsohn2022fixed} and $O\left((n/\epsilon)^{2/3}\right)$ by \cite{akrami2022efx}.  For four agents, \citet{BergerCoFe22} give a method that computes an EFX allocation while leaving at most one item unallocated.

Another related fairness notion of \EFP is called {\em local envy-freeness} (proposed by \citep{Beynier19}), in which each agent is only required to not envy her neighbors on an underlying social network. 
However, two agents connected by an edge are treated symmetrically with neither of them being prioritized.
On the other hand, strict envy-freeness is imposed from each prioritized agent to each non-prioritized agent in our notion \EFP.

\section{Preliminaries}
\label{sect:prelim}
A set $M$ of $m$ indivisible items is allocated to a set $N$ of $n$ agents.
Each agent $i$ has a \emph{valuation function} $v_i:\{0,1\}^m\to\mathbb{R}_{\geq0}$ that specifies a non-negative value to a bundle/set of items.\footnote{We will use the words ``bundle'' and ``set'' interchangeably.}
The valuation function $v_i$ is assumed to be
\begin{itemize}
    \item \emph{normalized}: $v_i(\emptyset)=0$; 
    \item \emph{monotone}: $v_i(S)\geq v_i(T)$ for any $T\subseteq S\subseteq M$.
\end{itemize}

A \emph{(complete) allocation} $(A_1,\ldots,A_n)$ is a partition of $M$, where $A_i$ is the set of items allocated to agent $i$.
A \emph{partial allocation} $(A_1,\ldots,A_n,B)$ is a partition of $M$ into $n+1$ subsets, where the extra subset $B$ is the set of unallocated items.
In this paper, unless specified otherwise, an allocation means a complete allocation.
We will only consider partial allocations in Sect.~\ref{sect:general}.

In a complete or partial allocation, we say that agent $i$ \emph{envies} agent $j$ if $v_i(A_i)<v_i(A_j)$.
That is, according to agent $i$'s utility function, agent $i$ believes her own bundle $A_i$ has less value than agent $j$'s bundle $A_j$.
An allocation is \emph{envy-free} if $i$ does not envy $j$ for any pair of agents $i$ and $j$.
An envy-free allocation may not exist in the problem of allocating indivisible items (e.g., when $m<n$). 
A well-known common relaxation of envy-freeness, \emph{envy-freeness up to one item} (EF1), is defined below.

\begin{definition}\label{def:EF1}
An allocation $(A_1, A_2, \dots, A_n)$ is said to satisfy \emph{envy-freeness up to one item (EF1)}, if for any two agents $i$ and $j$, there exists an item $g \in A_{j}$ such that $v_{i}(A_{i})\geq v_{i}(A_{j}\setminus\{g\})$.
\end{definition}
EF1 on partial allocations can be defined analogously.

For a verbal description, in an EF1 allocation, after removing some item $g$ from agent $j$'s bundle, agent $i$ will no longer envy agent $j$.
Given an allocation $(A_1,\ldots,A_n)$, we say that agent $i$ \emph{strongly envies} agent $j$ if $v_i(A_i)<v_i(A_j\setminus\{g\})$ for every $g\in A_j$.
By our definition, an allocation is EF1 if and only if $i$ does not strongly envy $j$ for every pair $(i,j)$ of agents.

As it is well-known that EF1 allocations always exist~\cite{Lipton04onapproximately}, strong envy can always be eliminated.
However, as mentioned before, envy may be unavoidable, and in this case we would like to give priority to a specified subset $P$ of agents such that each agent in $P$ does not envy each agent in $Q:= N\setminus P$.

\begin{definition}\label{def:EFP}
An allocation $(A_1,\ldots,A_n)$ is \emph{envy-free with respect to the set of prioritized agents $P$} if 1) the allocation is EF1, and 2) $i$ does not envy $j$ for any $i\in P$ and $j\in Q=N\setminus P$.
We say ``\EFP with respect to $P$'' to refer to this, or simply \EFP when the context is clear.
\end{definition}
This definition also applies to partial allocations.

\subsection{Additive Valuations}
\label{sect:additive}
To warm up, we first consider a common important special case of general valuations called \emph{additive valuations}: a valuation function $v_i$ is \emph{additive} if
$v_i(S)=\sum_{g\in S}v_i(\{g\})$.

In this case, an \EFP allocation for any subset $P\subseteq N$ can be simply achieved via the \emph{round-robin} algorithm with a specific agent order. 
In the round-robin algorithm, agents first queue up, then every agent from the top of the queue to the bottom takes away her favorite item among the remaining items and adds it to her bundle in turn.
EF1 is maintained during the round-robin algorithm.
In every round, an agent $i$ will not envy an agent $j_1$ if $j_1$ is behind $i$ in the queue.
For $j_2$ who is ahead of $i$, $i$ will not envy the item that $j_2$ receives in the next round.
Thus, when all the items have been allocated, after removing the item that $j_2$ receives in the first round, $i$ will not envy $j_2$.

In our algorithm, the agents in $P$ are ordered before the agents in $Q$.
It is easy to see that each agent $i \in P$ would pick an item before agent $j \in Q$ in every single round. Thus, $i$ will not envy $j$ when all the items have been allocated.\par 


Unfortunately, the above round-robin method cannot be extended to the case of general valuations, even for submodular valuations. The counterexample is shown as follows.
\begin{itemize}
    \item $P=\{p\}$ and $Q=\{q\}$;
    \item $M=\{1,2,3,4\}$;
    \item For agent $p$, $v_p(\{1\})=4$ and $v_p(\{2\})=v_p(\{3\})=v_p(\{4\})=3$. For the $2$-size set, $v_p(\{1,3\})=5$ and $v_p(\{1,2\})=v_p(\{1,4\})=4$, while $v_p(\{2,4\})=6$. The remaining valuations in $v_p$ can be chosen arbitrarily satisfying the submodular condition;
    \item For agent $q$, $v_q$ is additive where $v_q(\{1\})=v_q(\{3\})=0$ and $v_q(\{2\})=v_q(\{4\})=1$;
\end{itemize}

In the above example, all valuations are submodular.
In the round-robin algorithm, agent $p$ receives item $1$ and $3$ while agent $q$ receives item $2$ and $4$, where $p$ envies $q$ and violates the definition of \EFP.

\subsection{Envy-Graph}
\label{sect:prelim-cycle}
\cite{Lipton04onapproximately} first proposed the tool of \emph{envy-graph} for finding an EF1 allocation on general valuations. In an envy-graph, each vertex corresponds to an agent, and each directed edge $(u, v)$ represents that agent $u$ envies agent $v$. 
The \emph{envy-graph procedure} to find an EF1 allocation works as follows: it starts with an empty allocation and adjusts the allocation so that the envy-graph is a directed acyclic graph (DAG) before allocating the next item. When there is still an unallocated item, 
\begin{itemize}
    \item choose an arbitrary source agent (Definition~\ref{def:source}) of the envy-graph (so no one envies her), and allocate an item to her;
    \item reconstruct the graph according to the new allocation;
    \item If there exists a cycle in the envy-graph, we run cycle-elimination algorithm (defined in Definition~\ref{def:cycleElimination}). 
    The value for each agent is non-decreasing throughout the process after which the envy-graph contains no cycle.
\end{itemize}

It is easy to verify that the (partial) allocation is EF1 throughout the entire procedure.
In particular, adding an item $g$ to a source agent $i$ does not destroy the EF1 property, as an agent will not envy $i$ if $g$ is removed from $i$'s bundle.
The cycle-elimination step does not destroy the EF1 property either: this step does not change the constituents of each bundle, and each agent receives a bundle with a weakly larger value.

\begin{definition}[Cycle-Elimination]\label{def:cycleElimination}
For a cycle $u_0\rightarrow\dots\rightarrow u_{k-1}\rightarrow u_0$ on the envy-graph, each agent $u_i$ receives the bundle from $u_{i+1}$ where $i\in \{0, 1, \dots, k-1\}$ (indices are modulo $k$).
This is done iteratively until the envy-graph contains no cycle.
\end{definition}

We will also use the above cycle-elimination algorithm as a subroutine multiple times.
The cycle-elimination requires less than $n^2$ iterations, since each iteration removes at least one edge (the edges in the cycle are removed) and there are less than $n^2$ edges.
Thus, the envy-graph algorithm always terminates since the number of unallocated items is reduced by one after each iteration.

Lastly, we define a few notions that are used.

\begin{definition}\label{def:source}
An agent is called a \emph{source agent} if her corresponding vertex is a source in the envy-graph.
\end{definition}

\begin{definition}\label{def:Psource}
An agent $i$ is called a \emph{$P$-source agent} if she is a source agent in the subgraph induced by $P$. Moreover, $i$ is a \emph{$P$-source agent of agent $j$} if she is a $P$-source agent and $j$ is reachable from $i$ in the subgraph.
\end{definition}

\begin{definition}\label{def:Qsource}
An agent $i$ is called a \emph{$Q$-source agent} if she is a source agent in the envy-graph and $i\in Q$. Moreover, $i$ is a \emph{$Q$-source agent of agent $j$} if she is a $Q$-source agent and $j$ is reachable from $i$ in the envy-graph.
\end{definition}

In the envy-graph corresponding to a (partial) \EFP allocation, there is no edge from a vertex in $P$ to a vertex in $Q$.
Therefore, a $P$-source agent may not be a source agent, while a $Q$-source agent, by definition, is a source agent.
In addition, if $i$ is a $P$-source agent of $j$, we must have $j\in P$; if $i$ is a $Q$-source agent of $j$, then $j$ may be in $P$ or in $Q$.


\section{Identical Valuations of $P$ or $Q$}
\label{sect:specialgeneral}
This section studies two special cases of general valuations: when agents in $P$ have the same valuation, and when agents in $Q$ have the same valuation.
We design a polynomial-time algorithm that computes an \EFP allocation for each of the two cases.
The two algorithms in this section provide some basic ideas for our algorithm in Sect.~\ref{sect:general} where we consider general valuations.
In the next section, we will see how to extend these ideas to general valuations and the limitations of them.

Our results for the two special cases are also interesting on their own.
In the applications where $P$ is the underrepresented agents or $Q$ is the over-represented agents, it is natural to assume people in the same group (underrepresented or over-represented) share a similar valuation.
In addition, our results immediately apply to the three natural settings: when $P=1$, when $Q=1$, and when there are $3$ agents.


We first consider the case where all agents in $P$ have identical valuations by using Algorithm \ref{alg:identicalp}.
Lines \ref{algstep:idenpcase1b}-\ref{algstep:idenpcase1e} consider the simple case where an item can be directly assigned to one source agent in $P$, while Lines \ref{algstep:idenpcase2b}-\ref{algstep:idenpcase2e} try to assign an item to a suitable $Q$-source agent and keep no envy from $P$ to $Q$ by eliminating one envy cycle.
Lines \ref{algstep:idenpupb}-\ref{algstep:idenpupe} update the item set $B$ and the envy-graph after adding one item.

\begin{algorithm}[htbp]
\caption{Algorithm for computing \EFP allocation satisfying Theorem \ref{thm:identicalp} }\label{alg:identicalp}
\KwOutput{an \EFP allocation.}
Let $A_i = \emptyset$, for any $i\in N$, and $B = M$\;
\While{there exists an item $g\in B$}
{
    \If{there exists one source agent $i\in P$ \label{algstep:idenpcase1b}}
    {
     $A_i\leftarrow A_i\cup \{g\}$\; \label{algstep:idenpcase1e}
    }
    \Else 
    {
    Let $i$ be one $P$-source agent\;\label{algstep:idenpcase2b}
    Let $j$ be one $Q$-source agent of agent $i$\;\label{algstep:idenpcase2c}
    $A_j\leftarrow A_j\cup \{g\}$\;
    \If {agent $i$ envies $j$}
    {
        Add the edge $(i,j)$ to the envy-graph\;
        Eliminate the cycle $j\rightarrow\cdots\rightarrow i\rightarrow j$ \hspace{0.5cm}// there is a path from $j$ to $i$ by Definition~\ref{def:Qsource} and Line~\ref{algstep:idenpcase2c}\;\label{algstep:idenpcase2e}
    }
    }
    $B\leftarrow B\setminus \{g\}$\; \label{algstep:idenpupb}
    Reconstruct the envy-graph\;
    Run cycle-elimination algorithm (Definition~\ref{def:cycleElimination})\;\label{algstep:idenpupe}
}
\Return{$A$}
\end{algorithm}

\begin{theorem}
An \EFP allocation always exists and can be found in polynomial time when all agents in $P$ have identical valuations.
\label{thm:identicalp}
\end{theorem}

\begin{proof}
The cycle-elimination step at Line~\ref{algstep:idenpupe} always makes the envy-graph a DAG, which validates Line~\ref{algstep:idenpcase2b} and \ref{algstep:idenpcase2c}.
It can then be easily checked that Algorithm \ref{alg:identicalp} always terminates in polynomial time.
It suffices to show the allocation returned by Algorithm \ref{alg:identicalp} is \EFP. Since the empty allocation at the beginning is a partial \EFP allocation, we just need to show the allocation is still a partial \EFP allocation after each iteration of the while-loop.

We first consider the simple case where we can find one source agent $i\in P$ (Lines \ref{algstep:idenpcase1b}-\ref{algstep:idenpcase1e}).
Before allocating item $g$, there is no envy to agent $i$ and no envy from $P$ to $Q$ from the definition of source agent and partial \EFP allocation.
Thus, after allocating $g$ to agent $i$, there is no strong envy to agent $i$ (as the envy is eliminated if $g$ is removed from $i$'s bundle), and there is no envy from $P$ to $Q$ by the monotonicity of valuation functions. This is still a partial \EFP allocation.

We then consider the case where there is no source agent in $P$ (Lines \ref{algstep:idenpcase2b}-\ref{algstep:idenpcase2e}).
We first denote $i$ as one $P$-source agent and $j$ as one $Q$-source agent of $i$. As mentioned, since the envy-graph is a DAG, there must exist such agents $i$ and $j$.
After allocating item $g$ to agent $j$, two cases may happen:
\begin{itemize}
    \item Agent $i$ does not envy $j$. Since all agents in $P$ have identical valuations and $i$ is a $P$-source agent, we have $u_k(A_k)\geq u_i(A_i)\geq u_i(A_j\cup\{g\})=u_k(A_j\cup\{g\})$ for each agent $k\in P$, so there is no envy from $P$ to $Q$. There is no strongly envy to agent $j$ because there is no envy to agent $j$ before allocating item $g$.
    \item Agent $i$ envies $j$. We will eliminate the cycle including $j$ and $i$.  Since $i$ is a $P$-source agent and the only agent in $P$ whose bundle is reallocated to an agent in $Q$ is agent $i$, there is no envy from $P$ to $Q$.
    Because the bundle $A_j$ is held by one source agent before allocating $g$ and each agent in the cycle gets better bundle, there is still no strongly envy. 
\end{itemize}

In both two cases above, the allocation is still a partial \EFP allocation.
From the fact that there is no envy edge from $P$ to $Q$, there is no envy cycle containing agents in $P$ and $Q$ at the same time, so Line \ref{algstep:idenpupb}-\ref{algstep:idenpupe} keep the envy-graph acyclic and maintain \EFP.
\end{proof}

We then consider the case where all agents in $Q$ have identical valuations. The technique is similar to Algorithm~\ref{alg:identicalp}, where the only difference is that when there exists no source agent in $P$, we find and eliminate one envy cycle more carefully.

\begin{theorem}
An \EFP allocation always exists and can be found in polynomial time when all agents in $Q$ have identical valuations.
\label{thm:identicalq}
\end{theorem}
\begin{proof}
\begin{algorithm}[htbp]
\caption{Algorithm for computing \EFP allocation satisfying Theorem~\ref{thm:identicalq} }\label{alg:identicalq}
\KwOutput{an \EFP with respect to $P$ allocation.}
Let $A_i = \emptyset$, for any $i\in N$, and $B = M$\;
\While{there exists an item $g\in B$}
{
    \If{there exists one source agent $i\in P$ \label{algstep:idenqcase1b}}
    {
     $A_i\leftarrow A_i\cup \{g\}$\;\label{algstep:idenqcase1e}
    }
    \Else
    {
    Let $i$ be one source agent\;\label{algstep:idenqcase2b}
    $A_i\leftarrow A_i\cup \{g\}$\;
    \If{there exists $j\in P$ who envies $i$}
    {
    Add the edge $(j,i)$ to the envy-graph\;
    Let $k$ be one $P$-source agent of $j$ and $l$ be one $Q$-source agent of $k$\;
    Eliminate the cycle $i\rightarrow l\rightarrow\cdots\rightarrow k \rightarrow \cdots \rightarrow j \rightarrow i$, where the edge $i\rightarrow l$ is an added edge if $i\neq l$\;\label{algstep:idenqcase2e}
    }
    }
    $B\leftarrow B\setminus \{g\}$\;\label{algstep:idenqupb}
    Reconstruct the envy-graph\;
    Eliminate the cycles (Definition~\ref{def:cycleElimination})\;\label{algstep:idenqupe}
}
\Return{$A$}
\end{algorithm}

The algorithm is given in Algorithm~\ref{alg:identicalq}.
The technique is similar to Algorithm~\ref{alg:identicalp}, where the only difference is that when there exists no source agent in $P$, we need to find and eliminate one envy cycle more carefully, which is implemented at Lines \ref{algstep:idenqcase2b}-\ref{algstep:idenqcase2e}.

With the similar analysis as in the proof of Theorem \ref{thm:identicalp}, it suffices to show the allocation can keep partial \EFP with respect to $P$ and the corresponding envy-graph can keep acyclic after adding each item $g$.

Since the use of Lines \ref{algstep:idenqcase1b}-\ref{algstep:idenqcase1e} and \ref{algstep:idenqupb}-\ref{algstep:idenqupe} is the same as the corresponding part in the proof of Theorem \ref{thm:identicalp}, the remaining is to show after implementing Lines \ref{algstep:idenqcase2b}-\ref{algstep:idenqcase2e}, there is no strongly envy between each pair of agents and no envy from $P$ to $Q$.

Since there is no source agent in $P$ and the envy-graph is a DAG, agent $i$ exists and is in $Q$. If there is no agent in $P$ who envies the bundle $A_i\cup\{g\}$, because $i$ is a source agent before, no one strongly envies $i$ and there is still no envy from $P$ to $Q$ after allocating $g$.

We then consider the second case where there exists one agent $j\in P$ who envies agent $i$.
Because the envy-graph is a DAG and no source agent is in $P$, agent $k$ and $l$ exist.
Since all agents in $Q$ have the same valuations and agent $i$ and $l$ are both $Q$-source agents, $u_i(A_l)=u_l(A_l)\geq u_l(A_i)=u_i(A_i)$.
After eliminating the cycle in Line \ref{algstep:idenqcase2e}, each agent gets a weakly better bundle.
With the fact that the bundle $A_i$ is held by a source agent before adding item $g$, there is no strong envy towards the bundle $A_i\cup\{g\}$.
Since only agent $k$'s bundle in $P$ is held by one agent in $Q$ and $k$ is a $P$-source agent, there is no envy from $P$ to $Q$.
\end{proof}

From the above two theorems, we can easily get the tractability of the following three natural special cases.

\begin{corollary}
An \EFP allocation always exists and can be found in polynomial time in any of the following settings: when $|P|=1$, when $|Q|=1$, and when $|N|\le3$.
\end{corollary}

\section{General Valuations}
\label{sect:general}
In this section, we consider general valuations with no constraint.
The main challenge of applying Algorithm~\ref{alg:identicalp} to the setting here is that vertices in $P$ may not be well-connected.
Consider the following scenario.
We have a total of four agents $1,2,3,4$ where $P=\{1,2\}$ and $Q=\{3,4\}$.
After a certain iteration, we have a partial allocation where the envy-graph only have two edges $(3,1)$ and $(4,2)$.
Moreover, for each remaining unallocated item $g$, it satisfies that 1) adding $g$ to $1$ or $2$ introduces strong envy, 2) adding $g$ to $3$ makes $2$ envy $3$, and 3) adding $g$ to $4$ makes $1$ envy $4$.
Then, the algorithm cannot continue with existing techniques in the previous section.
In particular, no cycle appears if we add $g$ to agent $3$ or $4$, and the partial allocation is no longer \EFP.
In the previous setting, the good connections between agents in $P$ ensures that we can always make a cycle appear by carefully selecting an agent to whom an item is added.

Nevertheless, we are able to obtain a slightly weaker result for general valuations.
We prove that such a partial \EFP allocation always exists: the number of unallocated items is less than the size of both $P$ and $Q$, and no one envies the unallocated bundle $B$.
We will call the set of the unallocated items the \emph{pool}.

\begin{theorem}\label{thm:efppar}
For any $P\subseteq N$, a partial \EFP allocation (with respect to $P$) that satisfies the following properties always exists.
\begin{itemize}
\label{thm:partialefp}
    \item $|B| < \min(|P|, |Q|)$, and
    \item $v_i(B) \le v_i(A_i)$ for all $i \in N$.
\end{itemize}
In addition, there is a polynomial-time algorithm that computes such an allocation.
\end{theorem}

The algorithm shares some similarities with the one proposed by~\citet{charity_soda}.
However, there are substantial differences in the analysis of the algorithm as the objective is changed from EFX to \EFP.
Moreover, our update rule $U_3$ that exchanges an agent's bundle with a set of unallocated items is more technically involved compared to its counterpart in Chaudhury et al.'s algorithm.
This additional techniques also make our algorithm run in polynomial time, whereas Chaudhury et al.'s algorithm runs in pseudo-polynomial time.

\subsection{The Main Algorithm}
The main algorithm is shown in Algorithm~\ref{alg:partialalg}. 
Each iteration of Algorithm~\ref{alg:partialalg} applies one of the four update rules defined in Algorithms~\ref{alg:u0_1},~\ref{alg:u2} and~\ref{alg:u3}.
After that, the envy-graph is reconstructed and cycles in the graph are eliminated.
We will prove that the \EFP property is secured after applying any of the four rules, and we will also prove that the cycle-elimination step does not destroy the \EFP property.

Next, we will give more detailed descriptions of the four update rules.
\begin{algorithm}[htbp]
\caption{Computing a partial \EFP allocation}\label{alg:partialalg}
\KwOutput{a partial \EFP allocation satisfying Theorem~\ref{thm:partialefp}.}
Let $A_i = \emptyset$, for any $i\in N$, and $B = M$\;
\While{there exists an applicable rule $U_\ell$\label{algline:general_while}}%
{
    $A, B \leftarrow U_\ell(A, B)$\;
    Reconstruct the envy-graph\;
    Run cycle-elimination algorithm (Definition~\ref{def:cycleElimination})\;\label{algline:cycle-elimination}
}
\Return{the partial \EFP allocation $A$}
\end{algorithm}
\begin{algorithm}[htbp]
\caption{The Update Rules $U_0$ and $U_1$.}\label{alg:u0_1}
\SetKwProg{Fn}{Function}{:}{}
\Fn{$U_0$(allocation A, pool B)}{
\textbf{Precondition:} There exist an item $g\in B$ and one source agent $i\in P$\;
Allocate $g$ to $i$: $A_i \leftarrow A_i \cup \{g\}$\;
Update pool: $B\leftarrow B\setminus \{g\}$\;
}
\SetKwProg{Fn}{Function}{:}{}
\Fn{$U_1$(allocation A, pool B)}{
\textbf{Precondition:} There exist an item $g\in B$ and one source agent $i\in Q$ such that allocating $g$ to $i$ would not cause any agent in $P$ to envy $i$\;
Allocate $g$ to $i$: $A_i \leftarrow A_i \cup \{g\}$\;
Update pool: $B\leftarrow B\setminus \{g\}$\;
}
\end{algorithm}

\paragraph{Rule $U_0$:} As shown in Algorithm~\ref{alg:u0_1}, $U_0$ just allocates an unallocated item to a source agent in $P$. Hence, it would not lead to any new envy edges from $P$ to $Q$. 

\paragraph{Rule $U_1$:} When there is no source agent in $P$, $U_1$ (shown in algorithm~\ref{alg:u0_1}) will allocate an item to a harmless agent ($i.e.,$ no agents in $P$ will envy her after allocating this item) in $Q$. 

\paragraph*{Rule $U_2$:} 
$U_2$ works in the scenarios where neither $U_0$'s precondition nor $U_1$'s is satisfied and the size of pool items is at least $\min(|P|, |Q|)$. Hence, there are no source agents in $P$, and, for any agent in $Q$, allocating an unallocated item to her will cause an agent in $P$ to envy her. 

In this case, we first choose an arbitrary $P$-source agent $p_0$, and then find one $Q$-source agent $q_0$ of $p_0$ (This is always possible as the precondition of $U_0$ fails). 
Then, the while-loop runs several iterations described as follows, until the envy-graph has a cycle.
We choose an arbitrary unallocated item and add it to $q_0$'s bundle.
Since the precondition of $U_1$ fails, there exists an agent $p_1'\in P$ that envies $q_0$'s bundle.
We add the edge $(p_1',q_0)$ corresponding to this envy relation (we do not update the entire envy-graph at this moment).
Let $p_1\in P$ be one $P$-source of $p_1'$, and let $q_1\in Q$ be one $Q$-source of $p_1$ (again, the existence of $q_1\in Q$ is due to the failure of the precondition of $U_0$).
We add another unallocated item to $q_1$'s bundle, and let $p_2'\in P$ be one agent that envies $q_1$'s bundle (the existence of $p_2'\in P$ is based on the failure of the precondition of $U_1$).
We keep carrying on this process until a cycle is found.
The cycle will be a union of multiple segments ``$\rightarrow q_i\rightarrow\cdots \rightarrow p_i\rightarrow\cdots\rightarrow p_i'\rightarrow$''.
Notice that each segment contains at least one agent in $P$ and at least one agent in $Q$.
The precondition of $U_2$ ensures the unallocated items will not run out before we find a cycle (formally proved in Proposition~\ref{prop:u2_while_termination}).



After the while-loop, we reallocate the bundles along this cycle. 
Figure~\ref{fig:u2_eg} provides an illustration of the update rule $U_2$.

Notice that we do not update the entire envy-graph within each iteration of the while-loop. Instead, we only add one edge $(p_{i+1}',q_i)$.
If the entire envy-graph were updated, there may be more cycles within agents in $P$ or agents in $Q$.
However, these cycles are handled at Line~\ref{algline:cycle-elimination} of Algorithm~\ref{alg:partialalg} after we have completed the update $U_2$. 

\paragraph*{Rule $U_3$:}
$U_3$ is mainly used to handle the case when neither $U_0$'s precondition nor $U_1$'s is satisfied and there exists an agent who envies the pool.

As shown in Algorithm~\ref{alg:u3}, we first find a subset $S\subseteq B$ such that there exists an agent $s$ that envies $S$, yet for all $S'\subsetneq S$, no agent envies $S'$. 
If $s\in P$, then let $p$ be one $P$-source agent of $s$ and $q_0$ be one $Q$-source agent of $p$. If $s\in Q$, then let $q_0$ be one $Q$-source agent of $s$. 
Next, we try to add an item to $q_0$, and perform a similar sequence of steps as the while-loop in $U_2$.
Two events may happen: a cycle is found, or a cycle is not found after all the unallocated items are allocated.
In the former case, we terminate $U_3$ and jump to Line~\ref{algline:U2_cycle} of Algorithm~\ref{alg:u2} (rule $U_2$).
In the latter case, we reallocate the bundles as follows. Consider the path $u_0 \rightarrow \cdots \rightarrow u_k$ from $q_i$ to $s$. We reallocate their bundles along the path and let each agent take her succeeding vertex's bundle.
For the first agent $u_0=q_i$, her original bundle becomes the pool of unallocated items after receiving her succeeding agent's bundle.
For the last agent $u_k=s$, she takes the bundle $S$.

\begin{figure}[t]
    \centering
    \includegraphics[width=0.6\textwidth]{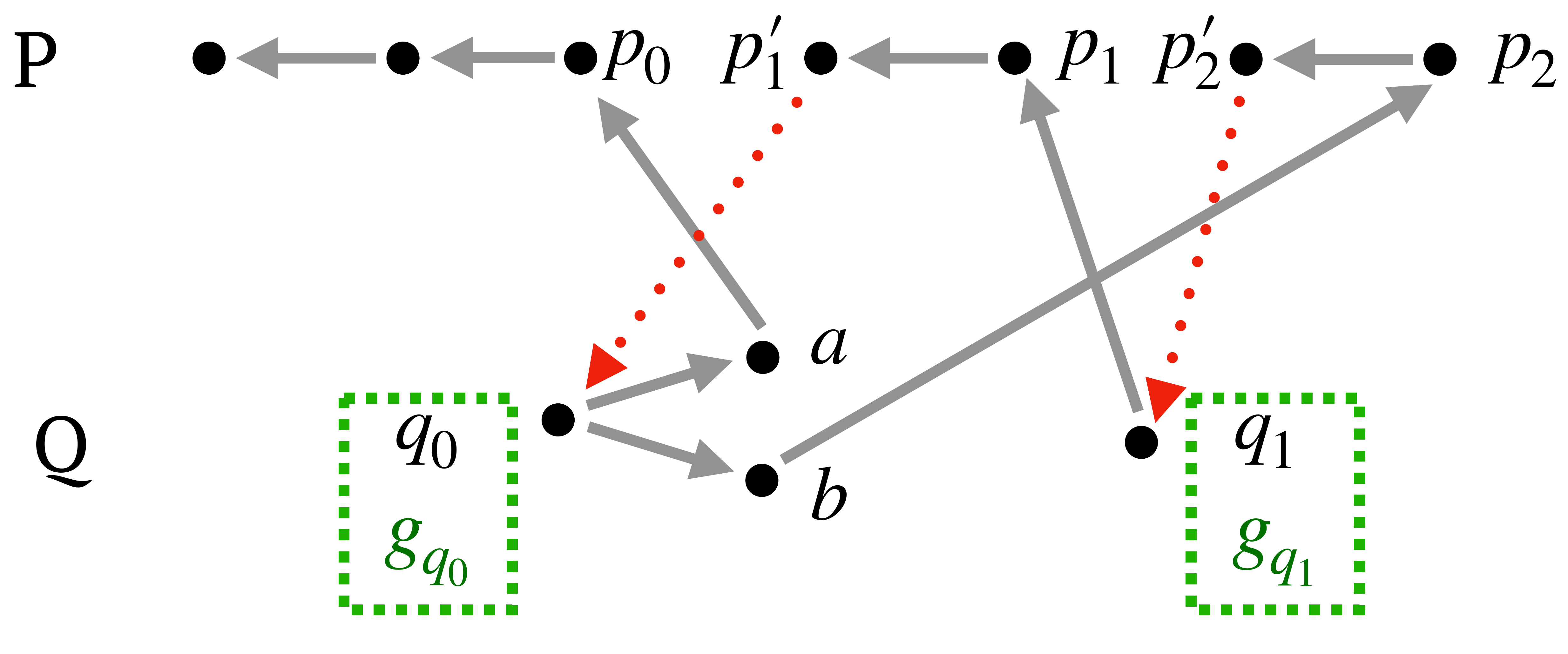}
    \caption{An example envy-graph for rule $U_2$. The solid lines and the dotted lines respectively represent the original edges and newly constructed edges during $U_2$. At the first iteration, item $g_{q_0}$ is allocated to agent $q_0$, and it causes agent $p_1'$ to envy $q_0$. Then we create an edge $(p_1', q_0)$, and let one $P$-source agent of $p_1'$ be $p_1$ and one $Q$-source agent of $p_1$ be $q_1$. Next, this loop terminates when finding one $Q$-source agent of $p_2$ is just the previous $q_0$. Finally, these agents form a cycle $q_0\rightarrow b \rightarrow p_2 \rightarrow p_2' \rightarrow q_1\rightarrow p_1\rightarrow p_1' \rightarrow q_0$. We remark that the cycle starts and ends at $q_0$ in this example, and the cycle may start at a ``middle vertex'' $q_i$ in general.}
    \label{fig:u2_eg}
\end{figure}
\begin{algorithm}[t]
\caption{The Update Rule $U_2$}\label{alg:u2} 
\SetKwProg{Fn}{Function}{:}{}
\Fn{$U_2$(allocation A, pool B)}{
\textbf{Precondition:} The preconditions of both $U_0$ and $U_1$ are not satisfied, and $|B|\ge \min(|P|, |Q|)$\;
Let $p_0\in P$ be an arbitrary $P$-source agent\;
Find one $p_0$'s $Q$-source agent $q_0\in Q$\;
$A' \leftarrow A, i\leftarrow 0$\;
\While{the envy-graph contains no cycle}{
Allocate $g_{q_i}$ to $q_i$: $A_{q_i}' \leftarrow A_{q_i}\cup\{g_{q_i}\}$\;\label{algline:U2_add}
Let $p_{i+1}'\in P$ be one agent who envies $A_{q_i}'$\;
Add the edge $(p_{i+1}', q_i)$ in the envy-graph\;\label{algline:U2_construt}
Let $p_{i+1}$ be one $P$-source agent of $p_{i+1}'$\;\label{algline:U2_p}
Let $q_{i+1}$ be one $Q$-source agent of $p_{i+1}$\;\label{algline:U2_q}
$i\leftarrow i+1$\;
}
Let $u_0\rightarrow\cdots\rightarrow u_{k-1}\rightarrow u_0$ be the cycle consisting of the segments ``$q_i\rightarrow\cdots\rightarrow p_i\rightarrow \cdots \rightarrow p_i'$''\;\label{algline:U2_cycle}
$A_{u_i}\leftarrow A_{u_{i+1}}'$ for each $i$ (indices are modulo $k$)\;\label{algline:U2_reallocate}
Update the pool $B$: $B\leftarrow M\setminus\left(\bigcup_{i=1}^nA_i\right)$\;
}
\end{algorithm}

\subsection{Proof of Theorem~\ref{thm:partialefp}}
In this section, we show that the output allocation of Algorithm~\ref{alg:partialalg} satisfies all the requirements in Theorem~\ref{thm:efppar}.

Below, we prove two properties for the update rule $U_2$.
The first proposition follows straightforwardly from Algorithm~\ref{alg:u2}.
\begin{restatable}{proposition}{utwopsource}\label{prop:u2_p_source}
For an edge $(i,j)$ in the cycle at Line~\ref{algline:U2_cycle} of the update rule $U_2$ (Algorithm~\ref{alg:u2}), the followings are true.
\begin{enumerate}
    \item If $i\in Q$ and $j\in P$, then $j$ is a $P$-source agent.
    \item If $j$ is an agent whose bundle has been updated at Line~\ref{algline:U2_add}, then $i\in P$. 
\end{enumerate}
\end{restatable}
\begin{proof}
Both statements hold straightforwardly from the update rule.
For each $p_i'$, we always find one $P$-source agent $p_i$ before back-tracking one source $q_i$ (Line~\ref{algline:U2_p}, \ref{algline:U2_q}), so an edge from $Q$ to $P$ must end at a $P$-source, which proves 1.
For 2, since we only update some $Q$-source agents' bundle at Line~\ref{algline:U2_add}, $j$ must be a $Q$-source agent. In addition, the edges pointing to $Q$-source agents are only constructed at Line~\ref{algline:U2_construt}, so $i$ must belong to $P$.
\end{proof}

The second proposition justifies the validity of the update rule $U_2$. 
\begin{proposition}\label{prop:u2_while_termination}
The while-loop in the update rule $U_2$ (Algorithm~\ref{alg:u2}) will terminate before the unallocated items running out.
\end{proposition}
\begin{proof}
We have seen that a cycle at Line~\ref{algline:U2_cycle} consists of multiple segments. For each segment, it is composed of three parts:  $q_i\rightarrow \cdots \rightarrow p_i$, $p_i\rightarrow \cdots\rightarrow p_i'$ and $p_i'\rightarrow q_{i-1}$ ($q_{i-1}$ is not included). Since each segment contains at least one agent belonging to $P$ and one agent belonging to $Q$, the number of segments is at most $\min(|P|, |Q|)$. In addition, since the first agent of each segment will receive exactly one new item from $B$, the number of segments is at most $\min(|P|, |Q|)$, which is no more than the number of unallocated items (guaranteed by the precondition of $U_2$). Therefore, the proposition concludes.
\end{proof}

\begin{algorithm}[t]
\caption{The Update Rule $U_3$.}\label{alg:u3}
\SetKwProg{Fn}{Function}{:}{}
\Fn{$U_3$(allocation A, pool B)}{
\textbf{Precondition:} The preconditions of both $U_0$ and $U_1$ are not satisfied, and there exists one agent that envies the unallocated bundle $B$\;
$B'\leftarrow B, S\leftarrow \emptyset, i\leftarrow 0, A'\leftarrow A$\;
\While{no agent envies $S$}{
Add an item $g\in B'$ to $S$\;
$B'\leftarrow B\setminus\{g\}$\;
}
\If{there exists $s\in P$ that envies $S$}{
Let $p$ be one $P$-source agent of $s$\;
Let $q_0$ be one $Q$-source agent of $p$\;
}\Else{
Let $s$ be one agent in $Q$ that envies $S$\;
Let $q_0$ be one $Q$-source agent of $s$\;
}
\While{$B'\neq\emptyset$}{
Allocate $g_{q_i}$ to $q_i$: $A_{q_i}' \leftarrow A_{q_i}\cup\{g_{q_i}\}$\;
Let $p_{i+1}'\in P$ be one agent who envies $A_{q_i}'$\;
Add the edge $(p_{i+1}', q_i)$ in the envy-graph\;
Let $p_{i+1}$ be one $P$-source agent of $p_{i+1}'$\;
Let $q_{i+1}$ be one $Q$-source agent of $p_{i+1}$\;
$i\leftarrow i+1$\;
\If{the envy-graph contains a cycle}{
Terminate $U_3$ and apply $U_2$ from Line~\ref{algline:U2_cycle}\;\label{algline:U3cycle}
}
}
Let $u_0\rightarrow\cdots\rightarrow u_k$ be the path from $q_{i}$ to $s$\;
Update pool: $B\leftarrow A_{q_{i}}$\;
$A_{u_i}\leftarrow A_{u_{i+1}}'$ for $i\in [1,k-1]$, $A_{u_k}\leftarrow S$\;
}
\end{algorithm}

The following proposition shows that the \EFP property is preserved after applying any of the four update rules.
\begin{proposition}\label{prop:four_rules}
For a partial \EFP allocation, the allocation is still \EFP after applying one iteration for any of $U_0$, $U_1$, $U_2$, and $U_3$.
\end{proposition}
\begin{proof}
For the first three rules, the allocation remains EF1 as the envy-graph procedure claims (see Sect.~\ref{sect:prelim-cycle}). Hence, we only need to prove that, for every agent in $p\in P$ and $q\in Q$, agent $p$ does not envy agent $q$. When rules $U_0$ and $U_1$ are applicable, it is easily checked by their preconditions.

We then analyze the rule $U_2$. 
Let $U = \{u_0, \ldots, u_m \}$ be the set of vertices in the cycle at Line~\ref{algline:U2_cycle} of Algorithm~\ref{alg:u2}. 
We discuss the following two cases of $q$. Note that in both cases, the value of agent $p$'s bundle would not decrease.
\begin{itemize}    \item $q\notin U$ : In this case, agent $q$ will still receive her old bundle. Hence, $p$ still will not envy $q$.
    \item $q\in U$: Agent $q$ would take the bundle from her adjacent agent (denote this agent as $r$) at Line~\ref{algline:U2_reallocate}. We consider two sub-cases: $r\in P$ and $r\in Q$.
    If $r\in P$, according to the first part of Proposition~\ref{prop:u2_p_source}, $r$ must be a $P$-source agent. In this case, no agent in $P$ envies $r$ before the reallocation (Line~\ref{algline:U2_reallocate}), so no agent in $P$ envies $q$ after the reallocation. 
    If $r\in Q$, according to the second part of Proposition~\ref{prop:u2_p_source}, $r$ cannot be an agent whose bundle has been updated at Line~\ref{algline:U2_add} as $q\in Q$. Since $p$ does not envy $r$ before applying $U_2$, $p$ will not envy $q$ after applying rule $U_2$. Thus, $p$ will not envy $q$ in both cases.
\end{itemize}

Now we come to rule $U_3$. 
Rule $U_3$ consists of two cases: the envy-graph forms a cycle at Line~\ref{algline:U3cycle}, or no cycle is formed and the unallocated items run out. 
In the first case, the bundle $S$ will get back to the pool, and the correctness is the same as rule $U_2$.
In the second case, the allocation is EF1 because all agents do not strongly envy $S$ (by our construction of $S$ with iterative addition of one item) and other items are added to the source agents such that each source agent receives at most one extra item. 
Next, we prove that any agent $p\in P$ will not envy agent $q\in Q$. 
It is worth noting that all agents' valuations to their own bundle will not decrease.
We consider the following three cases:
\begin{itemize}
    \item If $q$ is not on the path, her bundle stays the same and $p$'s valuation will not decrease. Hence, $p$ will not envy $q$.
    \item If $q=s$ where $s\in Q$, $p$ does not envy $q$. Otherwise, she will take the bundle $S$ according to Algorithm~\ref{alg:u3}.
    \item If $q\neq s$ is an agent on the path, during the reallocation process, she will receive a bundle from her adjacent agent, that is, an agent in $Q$ or a $P$-source agent. $p$ does not envy these two kinds of agents before the reallocation process and, after reallocation, her valuation to her bundle will not decrease, so she will not envy $q$.
\end{itemize}
Overall, the allocation is still \EFP after applying any of $U_0, U_1, U_2$ and $U_3$.
\end{proof}

The following proposition shows that the cycle-elimination at Line~\ref{algline:cycle-elimination} of our main algorithm (Algorithm~\ref{alg:partialalg}) does not destroy the \EFP property.
The proposition follows from that each cycle cannot contain vertices from both $P$ and $Q$ (since there is no edge from $P$ to $Q$).

\begin{restatable}{proposition}{cycleelimination}
\label{cycleelimination}
For a partial \EFP allocation, running cycle-elimination on the envy-graph does not violate the \EFP property.
\end{restatable}
\begin{proof}
Since there is no edge from $P$ to $Q$, the possible cycles in the envy-graph will involve vertices only in $P$ or only in $Q$. 
Since the cycle-elimination operation is a permutation of the previous allocation inside $P$ or $Q$ and everyone on the cycle gets a new bundle with higher valuation, the allocation is EF1 and no new edge occurs from $P$ to $Q$.
\end{proof}

Next, we show that the allocation output by our main algorithm is \EFP by induction. 
In the base case, no item is allocated to any agent, and it satisfies \EFP. 

Assume that after iteration $i$, the allocation satisfies \EFP and contains at least one $P$-source agent or $Q$-source agent (the cycle-elimination step at Line~\ref{algline:cycle-elimination} of Algorithm~\ref{alg:partialalg} ensures the existence of at least one source). In the following iteration, according to Proposition~\ref{prop:four_rules} and Proposition~\ref{cycleelimination}, \EFP will also not be destroyed after updating the allocation. Therefore, the inductive step concludes.

After showing that the output allocation is \EFP, the preconditions of rule $U_2$ and $U_3$, we conclude that the two requirements in Theorem~\ref{thm:partialefp} holds.
It remains to analyze the algorithm's time complexity.

The time complexity of Algorithm~\ref{alg:partialalg} is given below.
\begin{theorem}\label{thm:timecomplexitypartial}
The time complexity of Algorithm~\ref{alg:partialalg} is $O(n^2m\cdot\max(n^2,m))$.
\end{theorem}
\begin{proof}
Checking the preconditions for $U_0,U_1,U_2$, and $U_3$ are $O(1)$, $O(n^2m)$, $O(n^2m)$, and $O(n^2m)$ respectively.
The complexities for the four rules are $O(1)$, $O(1)$, $O(n^2)$, and $O(n^2)$ respectively.
The reconstruction of the envy-graph (after applying one of the update rules) takes $O(n^2)$ time.
For the cycle-elimination step, it takes $O(n^2)$ to find a cycle.
The cycle-elimination step requires at most $O(n^2)$ iterations, as each cycle-elimination eliminates at least one edge and there are at most $O(n^2)$ edges.
Thus, the cycle-elimination step takes $O(n^4)$ time.
The overall time complexity for one while-loop iteration of Algorithm~\ref{alg:partialalg} is $O(n^2\cdot\max(n^2,m))$.

We further claim that the while-loop at Algorithm~\ref{alg:partialalg}, Line~\ref{algline:general_while} is executed for a polynomial number of iterations.
The only case in the while loop that increases the size of the pool is when $U_3$ is applied and no cycle is formed.
The updated pool becomes a bundle from a previous $Q$-source agent, so no one envies the pool immediately after applying this rule. 
This operation may cause $|B|\ge \min\{|P|,|Q|\}$, and the while loop continues. 
However, during the application of the four rules, each agent's valuation to her bundle does not decrease, so no one will envy the pool anymore, and the precondition for $U_3$ will never be satisfied.
Hence, the case may appear only once and increase the size of the pool by at most $m$.

In other cases, each application of the rules will decrease the size of the pool by at least one. Hence, the while-loop will run for at most $2m$ iterations.

Hence, we may conclude that the time complexity of Algorithm~\ref{alg:partialalg} is $O(n^2m\cdot\max(n^2,m))$ which is a polynomial-time algorithm.
\end{proof}

\section{Conclusion and Open Problems}
In this paper, we studied fair division with prioritized agents. In particular, we proposed a new notion \EFP that is stronger than EF1 by allowing the allocation favors a prescribed subset of prioritized agents justified by some factors secondary to fairness. For general valuations, we proposed a polynomial-time algorithm that computes a partial \EFP allocation where the set of unallocated items has small values to all agents and a small cardinality.
We believe the existence and the computational tractability of a complete \EFP allocation is an important open problem.

Other than the settings with infinitely divisible resources (i.e., cake-cutting) and indivisible items, the setting with mixed divisible and indivisible items has received significant attention recently~\citep{bei2021fair,bei2021maximin}. 
Another future direction is to extend our \EFP notion to this setting. 

\section*{Acknowledgments}
The research of Shengxin Liu was partially supported by the National Natural Science Foundation of China (Grant No. 62102117), by the Shenzhen Science and Technology Program (Grant No. RCBS20210609103900003), by the Department of Education of Guangdong Province (Innovative Research Program: 2022KTSCX214), and by the Guangdong Basic and Applied Basic Research Foundation, and was partially sponsored by CCF-Huawei Populus Grove Fund.
The research of Biaoshuai Tao was supported by the National Natural Science Foundation of China (Grant No. 62102252).

\bibliographystyle{plainnat}
\bibliography{reference}

\begin{thebibliography}{36}
\providecommand{\natexlab}[1]{#1}
\providecommand{\url}[1]{\texttt{#1}}
\expandafter\ifx\csname urlstyle\endcsname\relax
  \providecommand{\doi}[1]{doi: #1}\else
  \providecommand{\doi}{doi: \begingroup \urlstyle{rm}\Url}\fi

\bibitem[Akrami et~al.(2022)Akrami, Chaudhury, Garg, Mehlhorn, and
  Mehta]{akrami2022efx}
Hannaneh Akrami, Bhaskar~Ray Chaudhury, Jugal Garg, Kurt Mehlhorn, and Ruta
  Mehta.
\newblock {EFX} allocations: {S}implifications and improvements.
\newblock \emph{arXiv preprint arXiv:2205.07638}, 2022.

\bibitem[Amanatidis et~al.(2022)Amanatidis, Aziz, Birmpas, Filos-Ratsikas, Li,
  Moulin, Voudouris, and Wu]{amanatidis2022fair}
Georgios Amanatidis, Haris Aziz, Georgios Birmpas, Aris Filos-Ratsikas, Bo~Li,
  Herv{\'e} Moulin, Alexandros~A Voudouris, and Xiaowei Wu.
\newblock Fair division of indivisible goods: A survey.
\newblock \emph{arXiv preprint arXiv:2208.08782}, 2022.

\bibitem[Aziz and Mackenzie(2016)]{AzizMa16}
Haris Aziz and Simon Mackenzie.
\newblock A discrete and bounded envy-free cake cutting protocol for any number
  of agents.
\newblock In \emph{Proceedings of the Annual IEEE Symposium on Foundations of
  Computer Science (FOCS)}, pages 416--427, 2016.

\bibitem[Aziz et~al.(2020)Aziz, Moulin, and Sandomirskiy]{AzizMoSa20}
Haris Aziz, Herv\'{e} Moulin, and Fedor Sandomirskiy.
\newblock A polynomial-time algorithm for computing a pareto optimal and almost
  proportional allocation.
\newblock \emph{Operations Research Letters}, 48\penalty0 (5):\penalty0
  573--578, 2020.

\bibitem[Bei et~al.(2012)Bei, Chen, Hua, Tao, and Yang]{bei2012optimal}
Xiaohui Bei, Ning Chen, Xia Hua, Biaoshuai Tao, and Endong Yang.
\newblock Optimal proportional cake cutting with connected pieces.
\newblock In \emph{Proceedings of AAAI Conference on Artificial Intelligence
  (AAAI)}, pages 1263--1269, 2012.

\bibitem[Bei et~al.(2017)Bei, Chen, Huzhang, Tao, and Wu]{bei2017cake}
Xiaohui Bei, Ning Chen, Guangda Huzhang, Biaoshuai Tao, and Jiajun Wu.
\newblock Cake cutting: Envy and truth.
\newblock In \emph{Proceedings of the International Joint Conference on
  Artificial Intelligence (IJCAI)}, pages 3625--3631, 2017.

\bibitem[Bei et~al.(2021{\natexlab{a}})Bei, Li, Liu, Liu, and Lu]{bei2021fair}
Xiaohui Bei, Zihao Li, Jinyan Liu, Shengxin Liu, and Xinhang Lu.
\newblock Fair division of mixed divisible and indivisible goods.
\newblock \emph{Artificial Intelligence}, 293:\penalty0 103436,
  2021{\natexlab{a}}.

\bibitem[Bei et~al.(2021{\natexlab{b}})Bei, Liu, Lu, and Wang]{bei2021maximin}
Xiaohui Bei, Shengxin Liu, Xinhang Lu, and Hongao Wang.
\newblock Maximin fairness with mixed divisible and indivisible goods.
\newblock \emph{Autonomous Agents and Multi-Agent Systems}, 35\penalty0
  (2):\penalty0 1--21, 2021{\natexlab{b}}.

\bibitem[Berendsohn et~al.(2022)Berendsohn, Boyadzhiyska, and
  Kozma]{berendsohn2022fixed}
Benjamin~Aram Berendsohn, Simona Boyadzhiyska, and L{\'a}szl{\'o} Kozma.
\newblock Fixed-point cycles and {EFX} allocations.
\newblock In \emph{Proceedings of the International Symposium on Mathematical
  Foundations of Computer Science (MFCS)}, pages 17:1--17:13, 2022.

\bibitem[Berger et~al.(2022)Berger, Cohen, Feldman, and Fiat]{BergerCoFe22}
Ben Berger, Avi Cohen, Michal Feldman, and Amos Fiat.
\newblock Almost full {EFX} exists for four agents.
\newblock In \emph{Proceedings of the AAAI Conference on Artificial
  Intelligence (AAAI)}, pages 4826--4833, 2022.

\bibitem[Beynier et~al.(2019)Beynier, Chevaleyre, Gourv{\`e}s, Harutyunyan,
  Lesca, Maudet, and Wilczynski]{Beynier19}
Aur{\'e}lie Beynier, Yann Chevaleyre, Laurent Gourv{\`e}s, Ararat Harutyunyan,
  Julien Lesca, Nicolas Maudet, and Ana{\"e}lle Wilczynski.
\newblock Local envy-freeness in house allocation problems.
\newblock \emph{Autonomous Agents and Multi-Agent Systems}, 33\penalty0
  (5):\penalty0 591--627, 2019.

\bibitem[Brams and Taylor(1995)]{brams1995envy}
Steven~J Brams and Alan~D Taylor.
\newblock An envy-free cake division protocol.
\newblock \emph{The American Mathematical Monthly}, 102\penalty0 (1):\penalty0
  9--18, 1995.

\bibitem[Bu et~al.(2022)Bu, Li, Liu, Song, and Tao]{bu2022complexity}
Xiaolin Bu, Zihao Li, Shengxin Liu, Jiaxin Song, and Biaoshuai Tao.
\newblock On the complexity of maximizing social welfare within fair
  allocations of indivisible goods.
\newblock \emph{arXiv preprint arXiv:2205.14296}, 2022.

\bibitem[Budish(2011)]{budish2011combinatorial}
Eric Budish.
\newblock The combinatorial assignment problem: Approximate competitive
  equilibrium from equal incomes.
\newblock \emph{Journal of Political Economy}, 119\penalty0 (6):\penalty0
  1061--1103, 2011.

\bibitem[Caragiannis et~al.(2019{\natexlab{a}})Caragiannis, Gravin, and
  Huang]{CaragiannisGrHu19}
Ioannis Caragiannis, Nick Gravin, and Xin Huang.
\newblock Envy-freeness up to any item with high {N}ash welfare: The virtue of
  donating items.
\newblock In \emph{Proceedings of the ACM Conference on Economics and
  Computation (EC)}, pages 527--545, 2019{\natexlab{a}}.

\bibitem[Caragiannis et~al.(2019{\natexlab{b}})Caragiannis, Kurokawa, Moulin,
  Procaccia, Shah, and Wang]{CaragiannisKuMo19}
Ioannis Caragiannis, David Kurokawa, Herv\'{e} Moulin, Ariel~D. Procaccia,
  Nisarg Shah, and Junxing Wang.
\newblock The unreasonable fairness of maximum {N}ash welfare.
\newblock \emph{ACM Transactions on Economics and Computation (TEAC)},
  7\penalty0 (3):\penalty0 12:1--12:32, 2019{\natexlab{b}}.

\bibitem[Chakraborty et~al.(2021)Chakraborty, Igarashi, Suksompong, and
  Zick]{ChakrabortyIgSu21}
Mithun Chakraborty, Ayumi Igarashi, Warut Suksompong, and Yair Zick.
\newblock Weighted envy-freeness in indivisible item allocation.
\newblock \emph{ACM Transactions on Economics and Computation (TEAC)},
  9\penalty0 (3):\penalty0 18:1--18:39, 2021.

\bibitem[Chakraborty et~al.(2022)Chakraborty, Segal-Halevi, and
  Suksompong]{ChakrabortySeSu22}
Mithun Chakraborty, Erel Segal-Halevi, and Warut Suksompong.
\newblock Weighted fairness notions for indivisible items revisited.
\newblock In \emph{Proceedings of the AAAI Conference on Artificial
  Intelligence (AAAI)}, pages 4949--4956, 2022.

\bibitem[Chaudhury et~al.(2020)Chaudhury, Garg, and Mehlhorn]{ChaudhuryGaMe20}
Bhaskar~Ray Chaudhury, Jugal Garg, and Kurt Mehlhorn.
\newblock {EFX} exists for three agents.
\newblock In \emph{Proceedings of the ACM Conference on Economics and
  Computation (EC)}, pages 1--19, 2020.

\bibitem[Chaudhury et~al.(2021{\natexlab{a}})Chaudhury, Garg, Mehlhorn, Mehta,
  and Misra]{ChaudhuryGaMe21}
Bhaskar~Ray Chaudhury, Jugal Garg, Kurt Mehlhorn, Ruta Mehta, and Pranabendu
  Misra.
\newblock Improving {EFX} guarantees through rainbow cycle number.
\newblock In \emph{Proceedings of the ACM Conference on Economics and
  Computation (EC)}, pages 310--311, 2021{\natexlab{a}}.

\bibitem[Chaudhury et~al.(2021{\natexlab{b}})Chaudhury, Kavitha, Mehlhorn, and
  Sgouritsa]{charity_soda}
Bhaskar~Ray Chaudhury, Telikepalli Kavitha, Kurt Mehlhorn, and Alkmini
  Sgouritsa.
\newblock A little charity guarantees almost envy-freeness.
\newblock \emph{SIAM Journal on Computing}, 50\penalty0 (4):\penalty0
  1336--1358, 2021{\natexlab{b}}.

\bibitem[Chen et~al.(2013)Chen, Lai, Parkes, and Procaccia]{CL10}
Yiling Chen, John~K Lai, David~C Parkes, and Ariel~D Procaccia.
\newblock Truth, justice, and cake cutting.
\newblock \emph{Games and Economic Behavior}, 77\penalty0 (1):\penalty0
  284--297, 2013.

\bibitem[Conitzer et~al.(2017)Conitzer, Freeman, and Shah]{ConitzerFrSh17}
Vincent Conitzer, Rupert Freeman, and Nisarg Shah.
\newblock Fair public decision making.
\newblock In \emph{Proceedings of the ACM Conference on Economics and
  Computation (EC)}, page 629–646, 2017.

\bibitem[Even and Paz(1984)]{even1984note}
Shimon Even and Azaria Paz.
\newblock A note on cake cutting.
\newblock \emph{Discrete Applied Mathematics}, 7\penalty0 (3):\penalty0
  285--296, 1984.

\bibitem[Farhadi et~al.(2019)Farhadi, Ghodsi, Hajiaghayi, Lahaie, Pennock,
  Seddighin, Seddighin, and Yami]{FarhadiGhHa19}
Alireza Farhadi, Mohammad Ghodsi, MohammadTaghi Hajiaghayi, S\'{e}bastien
  Lahaie, David Pennock, Masoud Seddighin, Saeed Seddighin, and Hadi Yami.
\newblock Fair allocation of indivisible goods to asymmetric agents.
\newblock \emph{Journal of Artificial Intelligence Research}, 64\penalty0
  (1):\penalty0 1–20, 2019.

\bibitem[Fish(1993)]{fish1993reverse}
Stanley Fish.
\newblock Reverse racism, or how the pot got to call the kettle black.
\newblock \emph{Atlantic Monthly}, 272\penalty0 (5):\penalty0 128--136, 1993.

\bibitem[Foley(1967)]{Foley67}
Duncan~Karl Foley.
\newblock Resource allocation and the public sector.
\newblock \emph{Yale Economics Essays}, 7\penalty0 (1):\penalty0 45--98, 1967.

\bibitem[Goldman and Procaccia(2015)]{goldman2015spliddit}
Jonathan Goldman and Ariel~D Procaccia.
\newblock Spliddit: Unleashing fair division algorithms.
\newblock \emph{ACM SIGecom Exchanges}, 13\penalty0 (2):\penalty0 41--46, 2015.

\bibitem[Kaiser et~al.(2013)Kaiser, Major, Jurcevic, Dover, Brady, and
  Shapiro]{kaiser2013presumed}
Cheryl~R Kaiser, Brenda Major, Ines Jurcevic, Tessa~L Dover, Laura~M Brady, and
  Jenessa~R Shapiro.
\newblock Presumed fair: Ironic effects of organizational diversity structures.
\newblock \emph{Journal of Personality and Social Psychology}, 104\penalty0
  (3):\penalty0 504, 2013.

\bibitem[Li et~al.(2022)Li, Bei, and Yan]{li2022proportional}
Zihao Li, Xiaohui Bei, and Zhenzhen Yan.
\newblock Proportional allocation of indivisible resources under ordinal and
  uncertain preferences.
\newblock In \emph{Proceedings of the Conference on Uncertainty in Artificial
  Intelligence (UAI)}, 2022.

\bibitem[Lipton et~al.(2004)Lipton, Markakis, Mossel, and
  Saberi]{Lipton04onapproximately}
Richard Lipton, Evangelos Markakis, Elchanan Mossel, and Amin Saberi.
\newblock On approximately fair allocations of indivisible goods.
\newblock In \emph{Proceedings of the ACM Conference on Electronic Commerce
  (EC)}, pages 125--131, 2004.

\bibitem[Newkirk and Vann(2017)]{newkirk2017myth}
Vann~R Newkirk and R~Vann.
\newblock The myth of reverse racism.
\newblock \emph{The Atlantic}, 5, 2017.

\bibitem[Plaut and Roughgarden(2020)]{PlautRo20}
Benjamin Plaut and Tim Roughgarden.
\newblock Almost envy-freeness with general valuations.
\newblock \emph{SIAM Journal on Discrete Mathematics}, 34\penalty0
  (2):\penalty0 1039--1068, 2020.

\bibitem[Shah(2017)]{shah2017spliddit}
Nisarg Shah.
\newblock Spliddit: {T}wo years of making the world fairer.
\newblock \emph{XRDS: Crossroads, The ACM Magazine for Students}, 24\penalty0
  (1):\penalty0 24--28, 2017.

\bibitem[Steinhaus(1948)]{Steinhaus1948}
Hugo Steinhaus.
\newblock The problem of fair division.
\newblock \emph{Econometrica}, 16\penalty0 (1):\penalty0 101--104, 1948.

\bibitem[Tao(2022)]{10.1145/3490486.3538321}
Biaoshuai Tao.
\newblock On existence of truthful fair cake cutting mechanisms.
\newblock In \emph{Proceedings of the ACM Conference on Economics and
  Computation (EC)}, page 404–434, 2022.

\end{thebibliography}

\end{document}